\documentclass[11pt]{article}
\usepackage{fullpage}
\usepackage{setspace}

\usepackage[backend=bibtex,style=alphabetic,backref=true,maxnames=4,minnames=1,maxbibnames=99]{biblatex}
\usepackage{amsmath,amssymb}
\usepackage{hyperref}
\hypersetup{
  colorlinks=true,
  citecolor=blue
}
\usepackage[amsmath,amsthm,thref,thmmarks,hyperref]{ntheorem}

\usepackage{mathtools}

\newtheorem{theorem}{Theorem}[section]
\newtheorem{lemma}[theorem]{Lemma}

\newtheorem{corollary}[theorem]{Corollary}

\theoremstyle{remark}

\theoremstyle{definition}
\newtheorem{definition}{Definition}

\usepackage{color}

\renewcommand{\vec}[1]{\boldsymbol{#1}}
\newcommand{\va}{\vec{a}}
\newcommand{\vb}{\vec{b}}

\newcommand{\ba}{\bar{a}}
\newcommand{\vba}{\vec{\ba}}
\newcommand{\vu}{\vec{u}}
\newcommand{\vx}{\vec{x}}
\newcommand{\bx}{\bar{x}}
\newcommand{\vbx}{\vec{\bx}}
\newcommand{\vy}{\vec{y}}
\newcommand{\eps}{\epsilon}

\newcommand{\poly}{\mathrm{poly}}

\newcommand{\calG}{\mathcal{G}}

\DeclareMathOperator*{\E}{\mathbb{E}}

\spacing{1.015}

\title{Well-Supported versus Approximate Nash Equilibria:\\ Query Complexity of Large Games\vspace{0.3cm}}

\author{
  Xi Chen\vspace{0.08cm}\hspace{0.03cm}\thanks{Email: {\tt xichen@cs.columbia.edu}. Supported in part by NSF grants CCF-1149257 and CCF-1423100.}
  \\ Columbia University
\and
  Yu Cheng\vspace{0.08cm}\hspace{0.03cm}\thanks{Email: {\tt yu.cheng.1@usc.edu}. Supported in part by NSF grant CCF-1111270 and Shang-Hua Teng's\newline Simons Investigator Award.}
  \\ University of Southern California
\and
  Bo Tang\vspace{0.08cm}\hspace{0.03cm}\thanks{Email: {\tt tangbonk1@gmail.com}. Supported in part by ERC grant 321171.}~\thanks{This work was done in part while the authors were visiting the Simons Institute.}
  \\ Oxford University
}
\date{}
\begin{document}
\maketitle\vspace{0.3cm}

\begin{abstract}
We study the randomized query complexity of approximate Nash equilibria (ANE)
  in large games.
We prove that, for some constant $\epsilon>0$, any randomized oracle algorithm~that
  computes an~$\epsilon$-ANE in a binary-action, $n$-player
  game 
  must~make $2^{\Omega(n/\log n)}$ payoff queries.
For the stronger solution concept of well-supported
  Nash equili\-bria (WSNE), Babichenko \cite{Babichenko:2014kd} previously gave
  an exponential $2^{\Omega(n)}$ lower~bound~for the randomized query complexity
  of $\epsilon$-WSNE, for some constant $\epsilon>0$;
  the same lower bound was shown~to hold for $\epsilon$-ANE, but
  only when $\epsilon=O(1/n)$.

Our result answers an open problem posed by Hart~and Nisan in \cite{hart2013query} 
  and by Babichenko~in \cite{Babichenko:2014kd},
  and is very close to the trivial upper bound of $2^n$.
Our proof relies on
  a generic reduction from the problem
  of finding an $\epsilon$-WSNE to the problem of finding
  an $\epsilon/(4 \alpha)$-ANE,~in large games with $\alpha$ actions,
  which might be of independent interest.
\end{abstract}


\clearpage

\section{Introduction}
\label{sec:intro}

The celebrated theorem of Nash \cite{NASH50} states that every finite game has an equilibrium point.~The solution concept of Nash equilibrium (NE) has been tremendously influential 
  in economics and social sciences ever since (e.g. see \cite{NashSocial}).
The complexity and efficient approximation of NE have been studied intensively during the past decade, and 
  much progress has~been made 
  (e.g., see~\cite{Lipton,AbbottKaneValiant,BaranyVempalaVetta,KannanTheobald,TSepsilon, 
  DGPJournal, 2Nash, FIXP, Mehta14, dask14,
  BPR15,Chen:2015ez,daskalakis2015structure,Rubinstein:2015kf, DiakonikolasKS15,Barman2015}).

In this paper, we study the randomized query complexity of 
  computing an $\epsilon$-\emph{approximate~Nash equilibrium} 
  (ANE) in large games, for some \emph{constant} $\epsilon>0$. 
Given a game $\calG$ with $n$~players and $\alpha$ actions for each player,  
  we index the players by the set $[n]=\{1,\ldots,n\}$ and
  index the actions by the set $[\alpha]=\{1,\ldots,\alpha\}$.
Recall that an $\epsilon$-ANE of $\calG$ is a mixed strategy profile 
$$
\vx=(\vx_1,\ldots,\vx_n),\quad \text{where $\vx_i\in [0,1]^\alpha$ sums to $1$ for each $i\in [n]$},
$$
in which each player $i$ plays an $\epsilon$-best response $\vx_i$ to other players' strategies
  $\vx_{-i}$\hspace{0.03cm}\footnote{We follow the convention and write
  $\vx_{-i}:=(\vx_1,\ldots,\vx_{i-1},\vx_{i+1},\ldots,\vx_n)$, strategies of players other than $i$ in $\vx$.} 
  (see the precise definition in \autoref{sec:prelim}).
Since the notion of ANE (as well as that of well-supported~Nash equilibria to be discussed below) is additive,
  we always assume that the payoff functions of games considered
  throughout this paper take values between $0$ and $1$.

For the (payoff) query model, an oracle algorithm with unlimited computational power
  is given an approximation parameter $\epsilon$, the number of players $n$ and
  the number of actions~$\alpha$~in an unknown game $\calG$, and needs to 
  find an $\epsilon$-ANE of $\calG$.
The algorithm has oracle access~to the payoff functions of players in $\calG$:
For each round, the algorithm can adaptively query
  a pure~strategy profile~$\va\in [\alpha]^n$,
  and receives the payoff of every player with respect to $\va$.
We~are interested~in the number of queries needed
  by any randomized oracle algorithm for this task.
Note that a trivial upper bound is $\alpha^n$,~by simply querying all the pure strategy profiles.



\subsubsection*{Prior Results and Related Work} 

The query complexity of (approximate) Nash equilibria and related solution concepts  has received considerable attention 
  recently, e.g. see \cite{Fearnley:2013,hart2013query,Fearnley:2014,goldberg2014query, Babichenko:2014kd,BB,
  GoldbergWINE}.~Below we~review results that are most relevant to our work.


The query complexity of (approximate) correlated equilibria\hspace{0.05cm}\footnote{An $\epsilon$-correlated
  equilibrium is a probability distribution over pure strategy profiles, i.e. $[\alpha]^n$,
  such that any player unilaterally deviating from strategies drawn from it can increase her
  expected payoff by no more than $\epsilon$.}
  (CE) is well understood.
For the \mbox{(payoff)} query model considered here,
  randomized algorithms exist (e.g., regret-mini\-mizing \mbox{algorithms} \cite{hart2000simple,hart2005adaptive,blum2007external})
  for finding an $\epsilon$-CE
  using $\poly(1/\epsilon, \alpha, n)$ many queries.
It turns out~that both randomization and approximation are necessary.
Babichenko and Barman in~\cite{BB} first showed that every deterministic algorithm that finds 
  an exact CE requires exponentially many queries in~$n$.
Hart and Nisan~\cite{hart2013query} then showed that the same exponential lower bound holds~for any deterministic algorithm for $(1/2)$-CE and any randomized algorithm for exact CE.
Now for the stronger (expected payoff) query model,
  where the oracle returns the expected payoffs
  of any mixed strategy profile\hspace{0.06cm}\footnote{Such an oracle can be implemented in polynomial time
  for many classes of succinct games; see \cite{papadimitriou2008correlated}.}, 
  Papadimitriou and Roughgarden~\cite{papadimitriou2008correlated} and
  Jiang and Leyton-Brown~\cite{jiang2011correlated} gave a deterministic algorithm
  that computes an exact CE in polynomial time using polynomially many queries (both in $\alpha$ and $n$).

Turning to the harder, but perhaps more interesting, problem
  of approximating Nash equilibria under~the payoff query model,
  the deterministic lower bound of \cite{hart2013query} for $(1/2)$-CE 
  directly implies the~same bound for $(1/2)$-ANE,
  since any $\eps$-ANE~by definition is also an $\eps$-CE.
For the randomized query complexity of NE,
  Babichenko~\cite{Babichenko:2014kd} showed that any randomized algorithm
  requires $2^{\Omega(n)}$ queries to find an $\epsilon$-well-supported Nash equilibrium
  (WSNE), in a binary-action, $n$-player game (see \autoref{thm:bab14}).
Recall that an $\epsilon$-WSNE of a game is a mixed strategy profile $\vx$
  in which the probability of player $i$ playing action $j$ is positive
  only when action $j$ is an $\epsilon$-best response with respect to $\vx_{-i}$ 
  (see \autoref{sec:prelim} for the precise definition).
By definition, an $\epsilon$-WSNE 
  is also an $\epsilon$-ANE but the inverse is not true. 
Following a well-known connection between WSNE and ANE \cite{DGPJournal}
  (and using random samples to approximate
  expected payoffs), Babichenko \cite{Babichenko:2014kd}~showed that~the same $2^{\Omega(n)}$
  bound holds for the randomized query complexity of $\epsilon$-ANE,
  but only when~$\epsilon=O(1/n)$. 
The randomized query complexity of $\epsilon$-ANE in large games,
  an arguably more natural relaxation~of exact NE compared to WSNE, 
  remains an open problem when $\epsilon>0$ is a constant. 

\subsubsection*{Our Results}

For binary-action, $n$-player games,
  we show that $2^{\Omega(n/\log n)}$ queries are required for any randomized 
  algorithm to find an $\epsilon$-ANE, for some constant $\epsilon>0$.
To state the result, we use $\text{\textsc{QC}}_p(\textbf{ANE}(n,\epsilon))$,  for 
  some $p>0$,  to denote the smallest 
  $T$ such that there exists a randomized oracle algorithm~that uses no more than $T$ queries
  and outputs
  an $\epsilon$-ANE with probability at least $p$,
  given any unknown binary-action, $n$-player 
  game. 
Our main result is the following lower bound on  $\text{\textsc{QC}}_p(\textbf{ANE}(n,\epsilon))$:
 
\begin{flushleft}\begin{theorem}[Main]\label{thm:main}
  There exist two constants $\epsilon>0$ and $c>0$ such that 
  $$\text{\textsc{QC}}_p(\textbf{\emph{ANE}}(n,\epsilon))=2^{\Omega(n/\log
    n)},\quad\text{where $p=2^{- c n/\log{n}}$.}$$
\end{theorem}\end{flushleft}

Our lower bound answers an open problem posed by Hart and Nisan \cite{hart2013query} 
  and by Babichenko \cite{Babichenko:2014kd}.
Our result shows that, in terms of their query complexities, finding
  an $\epsilon$-ANE is almost as hard as 
  finding an $\epsilon$-WSNE in a large game, even for constant $\epsilon>0$. 
It also implies the following corollary regarding the
  rate of convergence of $k$-queries dynamics (see \cite{Babichenko:2014kd} for
  the definition).
  
\begin{flushleft}\begin{corollary}
    There exist two constants $\epsilon>0$ and $c>0$ such that no
    $k$-queries dynamic can converge to an $\epsilon$-ANE in $\smash{2^{\Omega(n/\log n)}/k}$ steps with
    probability at least $2^{- {c n}/{\log n}}$ in all
    binary-action and $n$-player games.
\end{corollary}\end{flushleft}

In addition to the randomized query complexity, our proof of \thref{thm:main}
  yields a polynomial-time reduction\hspace{0.02cm}\footnote{Recall that a 
  polynomial-time reduction from total search problem $A$ to total search problem $B$
  is a pair $(f,g)$~of polynomial-time computable functions such that:
  1) for every input instance $x$ of $A$, $f(x)$ is an input instance of $B$;  
  2) for every solution $y$ to $f(x)$ in $B$, $g(y)$ is a solution to $x$ in $A$.} from the problem of finding an
  $\epsilon$-WSNE to that of finding an $(\epsilon'=\Omega(\epsilon))$-ANE
  in a \emph{succinct game} with a fixed number of actions.
Following the definition from
  \cite{papadimitriou2008correlated}, we say that~an $\alpha$-action succinct game is a pair $(n,U)$, where $n$
  is the number of players and
  $U$ is a (multi-output) Boolean circuit 
  that, given a pure strategy profile $\va\in [\alpha]^n$ (encoded in binary), outputs
  the payoffs of all $n$ players with respect to $\va$ in the game.
We show that

\begin{flushleft}
\begin{theorem}\label{thm:succinct}
Let $\epsilon\ge 0$ and $\alpha\in \mathbb{N}$ be two constants.
Then the problem of finding an $\epsilon$-WSNE 
  is polynomial-time reducible to that of finding an $\epsilon/(4\alpha)$-ANE, both in an $\alpha$-action succinct game.
\end{theorem}\end{flushleft}

\subsubsection*{Approximate vs Well-Supported Nash Equilibria}

Let $\text{\textsc{QC}}_p(\textbf{WSNE}(n,\epsilon))$,  for 
  some $p>0$, denote the smallest 
  $T$ such that there exists a randomized oracle algorithm that uses no more than $T$ queries
  and outputs
  an $\epsilon$-WSNE with probability at~least $p$,
  given any unknown binary-action, $n$-player game.
Babichenko showed that
\begin{theorem}[\cite{Babichenko:2014kd}]
\label{thm:bab14}
  There exist two constants $\eps>0$ and $c>0$ such that
  $$\text{\textsc{QC}}_p(\textbf{\emph{WSNE}}(n,\epsilon))=2^{\Omega(n)},\quad\text{where 
  $p=2^{-c n}$.}$$  
\end{theorem}

Given the result of Babichenko as above, the same exponential lower bound
  for the randomized query \mbox{complexity} of $\epsilon$-ANE,
  for small enough constant $\epsilon>0$, would follow immediately if 
\begin{flushleft}
\begin{enumerate} 
\item[] \emph{Given oracle access to $\calG$ and any $\epsilon'$-ANE of $\calG$, where
  $\epsilon'=c(\alpha)\cdot \epsilon$ for some constant $c>0$ that only depends on $\alpha$, 
  there is a query-efficient procedure that outputs an $\epsilon$-WSNE of $\calG$.}
\end{enumerate}
\end{flushleft}
However, the best such procedure known is the following result from \cite{DGPJournal}.
The parameters~are subsequently improved in \cite{Babichenko:2014kd},
  where the number of queries needed is also analyzed:
\begin{flushleft}\begin{quote}
\emph{Given oracle access to $\calG$ and any $\eps^2/(16n)$-ANE 
  of $\calG$, there is a procedure that outputs 
  an $\epsilon$-WSNE of $\calG$, where $n$ is the number of players,
  using $\text{poly}(\alpha,n,1/\epsilon)$ payoff queries.}
\end{quote}\end{flushleft}
The procedure is very natural: For each player, reallocate probabilities on
  actions with a relatively low expected payoff to a best-response action.
By \autoref{thm:bab14}, such a procedure implies~the~same exponential lower bound for $\epsilon$-ANE
  \cite{Babichenko:2014kd} but only when $\epsilon$ is $O(1/n)$.

No better procedure is known. By definition, an ANE
  poses a slightly weaker condition on~each player compared to that of a WSNE.
More specifically, given mixed strategies of other players $\vx_{-i}$,
  for an $\epsilon$-WSNE,
  $\vx_i$ must be supported on actions that are $\epsilon$-best responses to $\vx_{-i}$,
  while in an $\epsilon$-ANE, $\vx_i$ can be any
  mixed strategy that yields an overall $\epsilon$-best response to $\vx_{-i}$.
For example, $\vx_i$ may put $1-\epsilon$ probability on
  best-response actions while putting $\epsilon$ probability on any other actions.
On the one hand, this makes WSNE much easier to analyze
  and control in hardness reductions,
  which is why it played 
  a critical role in characterizing the complexity of Nash equilibria,
  starting with the work of \cite{DGPJournal}, later in \cite{2Nash} and subsequent works.
On~the~other~hand,~as the $\epsilon$ being of interest in \cite{DGPJournal,2Nash}
  is either exponentially or polynomially small, any hardness result for $\epsilon$-WSNE yields the same result for $\epsilon$-ANE (by combining the
  procedure of \cite{DGPJournal} described above and a folklore padding argument).



\subsubsection*{Our Approach}

While we were not able to improve the procedure of \cite{DGPJournal,Babichenko:2014kd},
  we prove \thref{thm:main}~via a \emph{query-efficient} reduction from
  the problem of finding a WSNE to that of 
  finding a ANE:
\begin{flushleft}\begin{quote}
\emph{Given any $\alpha$-action, $n$-player game $\calG$ and any parameter $\epsilon>0$,
  one can define a new $\alpha$-action game $\calG'$ with 
  a slightly larger set of $O(\alpha^2\log(n/\epsilon)\cdot n)$ players 
  such that}
\begin{flushleft}\begin{enumerate}
\item \emph{To answer each payoff query on $\calG'$, it suffices to make $\alpha n$
  payoff queries on $\calG$;}
\item \emph{There is a procedure that,  
given any $\epsilon$-ANE $\vx$ of $\calG'$, outputs
  a $(4\alpha\epsilon)$-WSNE $\vy$\\ of $\calG$, with no 
  payoff oracle access to $\calG$ or $\calG'$.}
\end{enumerate}\end{flushleft}\end{quote}\end{flushleft}
Our reduction is presented in \autoref{sec:reduc}.
\thref{thm:main} then follows directly from the lower bound~of Babichenko \cite{Babichenko:2014kd}
  on the randomized query complexity of WSNE (in \autoref{thm:bab14}).
\thref{thm:succinct} follows from the fact that:\hspace{-0.07cm} 1) the payoff entries of 
  $\calG'$ are easy to compute; and 2) the procedure to obtain
  $\vy$ from $\vx$ runs in  time polynomial in the length of
 the binary representation of $\vx$, when the number
  of actions $\alpha$ is bounded.





Recall that in the procedure of
\cite{DGPJournal,Babichenko:2014kd}, an $\epsilon$-WSNE is obtained 
  from an $\epsilon'$-ANE with $\epsilon'=\epsilon^2/(16n)$ 
  by reallocating probabilities on actions with
  relatively low expected payoff
  (formally, actions with payoff $\Omega(\epsilon)$ lower than the best response)
  to best-response actions.
From~the definition of ANE, no player can have probability more than
  $O(\eps'/\eps)=O(\eps/n)$ on actions with low payoff~in an $\epsilon'$-ANE.
Thus, the procedure changes the expected payoff of each player on each action
  by at most $n\cdot O(\eps/n)=O(\eps)$ 
  since it changes the mixed strategy of each player by $O(\eps/n)$.
  It follows~that the new mixed strategy profile
  is an $\epsilon$-WSNE.
The blow up of a factor of $n$
  from $\epsilon'$ to $\epsilon$ is precisely due to the cumulative
  impact on a player's expected payoff
  imposed by small changes~to all other players' mixed strategies.

Our reduction from WSNE to ANE overcomes this obstacle by constructing a new and
slightly larger game $\calG'$ with  $O(n\log{n})$ players, where each player $i$ in
the original $n$-player game $\calG$ is now simulated by a group of $O(\log{n})$
  players indexed by $(i,j)$ in the new game $\calG'$. 
The payoff function of player~$(i,j)$ in $\calG'$ 
  is exactly the same as that of player $i$ in $\calG$, but is defined
  with respect to the \emph{aggregate}~action of each group of players in $\calG'$
  by taking the \emph{majority} among each group.

We then show that an $\epsilon$-WSNE
  of $\calG$ can be recovered from an $\epsilon'$-ANE of $\calG'$,
  with $\epsilon'=\Omega(\epsilon)$,~by
  1) computing the distribution of the majority action of each group and
  2) truncating the small entries in each distribution.
Intuitively, by focusing on the aggregate behavior of each group
  of~$O(\log n)$ \emph{independent} players in $\calG'$,
  we make sure that the mixed strategies obtained from Step 1)
  are highly concentrated on actions with close-to-best expected payoffs,
  and actions with low payoffs can only appear as the majority action of a group
  with probability $O(\eps/n)$.
Therefore in Step 2), we only need to truncate entries with probability $O(\eps/n)$,
  and the remaining positive entries would correspond to close-to-best actions.
We can also control the effect of this truncation at the same time,
  because when the number of actions are bounded,
  the aggregate behavior of each group changes by at most $O(\eps / n)$.
This allows us to show that the result is an $\eps$-WSNE of the original game $\calG$.




\subsubsection*{Organization}
The rest of the paper is organized as follows.
We first give formal definitions of ANE and WSNE~in \autoref{sec:prelim}.
In \autoref{sec:reduc} we present the reduction from WSNE to ANE for large games,
  and then~use it to prove \autoref{thm:main} and \autoref{thm:succinct} in \autoref{sec:succinct}.
We conclude and discuss open problems in \autoref{sec:conclude}.



\section{Preliminaries}\label{sec:prelim}
A game $\calG$ is a triple $(n,\alpha,\vu)$, where $n$ is the number of players,
  $\alpha$ is the number of actions~for~each player,
  and $\vu=(u_1,\dots, u_n)$ are the payoff functions, one for each player.
We always use $[n]=\{1,\ldots,n\}$ to denote the set~of players and $[\alpha]=\{1,\ldots,\alpha\}$ to 
  denote the set of actions for each player.
Since we are interested in additive approximations, each $u_i$
  maps $[\alpha]^n$ to $[0,1]$.

Let $\Delta_\alpha$ denote the set of probability distributions over $[\alpha]$.
A mixed strategy profile of $\calG$ is then a tuple $\vx=(\vx_1,\dots,\vx_n)$ of
  mixed strategies, where $\vx_i \in \Delta_\alpha$ denotes the mixed strategy of player $i$.
Given $\vx$, we use $\vx_{-i}$ to denote the tuple of mixed strategies of all players other than $i$.
As a shorthand, we write $u_i(\vx)$ to denote the expected payoff of player $i$
  with respect to $\vx$, and write $u_i(a,\vx_{-i})$ to denote the expected payoff of player $i$
  playing action $a\in \big[\alpha\big]$ with respect to $\vx_{-i}$:
$$u_i(\vx)=\E_{\va \sim\vx }[u_i(\va)]\quad\text{and}\quad
u_i(a,\vx)=\E_{\vb \sim \vx_{-i}}\big[u_i(a,\vb)\big].
$$
  

Next we define approximate and well-supported Nash equilibria.

\begin{definition}
Given $\eps>0$, an $\epsilon$-approximate Nash equilibrium of an $\alpha$-action and
  $n$-player game $\calG(n,\alpha,\vu)$
  is a mixed strategy profile $\vx=(\vx_1,\dots,\vx_n)$ such that for every player $i\in [n]$:
  \[ u_i(\vx)\ge u_i(a',\vx_{-i})-\eps,\quad\text{for all $a'\in [\alpha]$.}\vspace{0.1cm} \]
\end{definition}

\begin{definition}
Given $\eps>0$, 
  an $\epsilon$-well-supported Nash equilibrium of $\calG(n,\alpha,\vu)$ is 
  a mixed strategy profile $\vx=(\vx_1,\dots,\vx_n)$ such that
  for every player $i\in [n]$ and every action $a$ in the support of $\vx_i$:
  \[ u_i(a, \vx_{-i})\ge u_i(a',\vx_{-i})-\eps,\quad\text{for all $a'\in [\alpha]$.}\vspace{0.1cm}\]
\end{definition}

Finally, we give a formal definition of succinct games \cite{papadimitriou2008correlated}. 

\begin{definition}
An $\alpha$-action succinct game is a pair $(n,U)$, where $n$
  is the number of players and~$U$ is a (multi-output) Boolean circuit that,
  given any pure strategy profile $\va\in [\alpha]^n$ (encoded in binary), outputs
  the payoffs of all $n$ players with respect to $\va$ in the game.
The input size of $(n,U)$ is the size of the circuit $U$. 
\end{definition}







\section{A Reduction from Well-Supported to Approximate\\ Nash Equilibria}\label{sec:reduc}

Given an $\alpha$-action, $n$-player game $\calG(n,\alpha,\vu)$ and $\epsilon\in (0,1)$,
  we define a new game $\calG'(sn,\alpha,\vu')$ with $sn$ players, where
$$
s= 2\alpha^2\cdot \left\lceil \ln (n/\epsilon) \right\rceil.
$$
We then show that, given an $\epsilon$-ANE $\vx$ of the new game $\calG'$, 
  one can compute a $(4\alpha \epsilon)$-WSNE $\vy$ of~$\calG$ 
  without making any payoff queries to $\calG$ or $\calG'$. 

For each player $i\in [n]$ in $\calG$, we introduce a group of $s$ players
  in $\calG'$, indexed by $(i,j)$ with $j\in [s]$,
  and use $u'_{i,j}$ to denote the payoff function of player $(i,j)$.
Given any pure strategy profile $\va=(a_{i,j}:i\in [n],j\in [s])$,
  we define the payoff $u'_{i,j}(\va)$ of player $(i,j)$ as follows.
First, for each $i\in [n]$,~let~$\ba_i\in [\alpha]$ denote the \emph{majority} action 
  played by the $i$-th group (players $(i,j)$, $j\in [s]$)
  in the pure strategy profile $\va$
  (break ties by choosing the action with the smallest index).
Write $\vba=(\ba_1,\ldots,\ba_n)$.
Next,~the payoff of player $(i,j)$ under $\va$ is defined as 
  \begin{equation}
    \label{eq:u-def2}
    u'_{i,j}(\va)=u_i(a_{i,j},\vba_{-i}).
  \end{equation}
This completes the definition of $\calG'$.
The lemma below follows directly from the definition.

\begin{lemma}\label{lem:queryefficient}
To answer a payoff query on $\calG'$, it suffices to make $\alpha n$ payoff queries on $\calG$.
\end{lemma}
\begin{proof}
By the definition of $\calG'$, $u'_{i,j}(\va)$'s for all $(i,j)$, $i\in [n]$ and $j\in [s]$, are
  determined by $$\big(u_i(a',\vba_{-i}):i\in [n],a'\in [\alpha]\big),$$ for which
  $\alpha n$ payoff queries on $\calG$ suffice.
\end{proof}

We conclude our reduction by proving the following lemma:

\begin{flushleft}\begin{lemma}
  \label{lem:WSNtoANE}
Given any $\epsilon$-ANE $\vx$ of $\calG'$, 
  one can compute a $(4\alpha\epsilon)$-WSNE $\vy$ of $G$ without making 
  any payoff queries on $\calG$ or $\calG'$.
Moreover, when $\alpha$ is a constant, 
  the computation of $\vy$ from $\vx$ can be done in time polynomial in the number of bits needed 
  in the binary representation of $\vx$ and $1/\epsilon$.
\end{lemma}\end{flushleft}
\begin{proof}
Let $\vx=(\vx_{i,j}:i\in [n],j\in [\alpha])$ be an $\epsilon$-ANE of $\calG'$.
For each group $i$ and action $k\in [\alpha]$, let
\begin{equation}\label{eq:bb}
\bx_{i,k}=\Pr_{\va\sim \vx}\left[\ba_i=k\right].
\end{equation}
Recall that $\ba_i$ is the majority action played by players $(i,j)$, $j\in [s]$, in 
  the pure strategy profile $\va$.
Then by definition, each $\vbx_i=(\bx_{i,1},\ldots,\bx_{i,\alpha})$ is a probability distribution over $[\alpha]$.

Next we define a mixed strategy $\vy=(\vy_1,\ldots,\vy_n)$ of $\calG$, and 
  show that $\vy$ is a $(4\alpha\epsilon)$-WSNE. Let
  \begin{equation}
    \label{eq:prob_truncate2}
    c_{i,k}=\left\{\begin{array}{ll}
      \bx_{i,k} & \textrm{ if } \bx_{i,k}\le {\epsilon}/{n} \\[0.5ex]
      0 & \textrm{ otherwise }
    \end{array}\right.
    \quad\text{and}\quad\ 
    y_{i,k}=\frac{\bx_{i,k}-c_{i,k}}{1-\sum_{j\in[\alpha]}c_{i,k}}.
  \end{equation}
It is clear that $\vy_i=(y_{i,1},\ldots,y_{i,\alpha})$ is a probability distribution over $[\alpha]$.

Now assume for contradiction that $\vy$ is not a $(4\alpha\epsilon)$-WSNE.
Then for some player $i\in [n]$,
  there exists an action $\ell\in [\alpha]$ such that $y_{i,\ell}>0$ but
\begin{equation}\label{eq:1}
\max_{k\in[\alpha]}\hspace{0.05cm}u_i(k,\vy_{-i})>u_i(\ell,\vy_{-i})+4\alpha\epsilon.
\end{equation}
But
  note that, the total variation distance between $\vbx_j$ and $\vy_j$ for each $j\in [n]$ is
  at most $\alpha\epsilon/n$. 
So by coupling and applying
  union bound, we have that
\begin{equation}\label{eq:2}
\big|u_i(k,\vbx_{-i})-u_i(k,\vy_{-i})\big|\le (n-1)\cdot (\alpha\epsilon/n)<\alpha\epsilon,\quad \text{for all
  $k\in[\alpha]$.}
\end{equation}
It then follows from (\ref{eq:1}) and (\ref{eq:2}) that
\begin{equation}\label{eq:3}
\max_{k\in[\alpha]}\hspace{0.05cm}u_i(k,\vbx_{-i})>u_i(\ell,\vbx_{-i})+2\alpha\epsilon.
\end{equation}  
By the definition 
  \eqref{eq:u-def2} of the payoff function $u'_{i,j}$, 
  we have 
\begin{equation}\label{eq:4}
u'_{i,j}(k,\vx_{-i})=u_i(k,\vbx_{-i}),\quad
  \text{for all $j\in[s]$ and $k\in[\alpha]$.}
\end{equation} 
Combining (\ref{eq:3}) and (\ref{eq:4}), we have that for every player $(i,j)$, $j\in [s]$:
$$
\max_{k\in[\alpha]}\hspace{0.05cm}u'_{i,j}(k,\vx_{-i})-u'_{i,j}(\ell,\vx_{-i})\ge
  2\alpha\epsilon.
$$
Since $\vx$ is an $\epsilon$-ANE of $\calG'$, 
  $x_{i,j,\ell}\le 1/(2\alpha)$. By Hoeffding bound
  and plugging in $s=2\alpha^2\cdot \lceil \ln(n/\epsilon)\rceil$,
\begin{align*}
  \bx_{i,\ell} &=\hspace{0.03cm} \Pr\big[\hspace{0.05cm}
  \text{$\ell$ is the majority action among players $(i,j)$, $j\in [s]$}\hspace{0.05cm}\big]\\[0.2ex]
    &\le \hspace{0.03cm}\Pr\big[\hspace{0.05cm}\text{the number of players $(i,j)$ playing $\ell$ 
    is at least $s/\alpha$}\hspace{0.05cm}\big]\le 
          e^{-{s}/({2\alpha^2})}\le \epsilon/n.
\end{align*}
  By~\eqref{eq:prob_truncate2}, this implies that $y_{i,\ell}=0$, which contradicts our assumption
  and proves that $\vy$ is indeed a $(4\alpha\eps)$-WSNE of $\calG$.
From the definition of $\vy$, it is clear that the computation of
  $\vy$ from $\vx$ does not require any payoff queries.
  
For the running time, when $\alpha$ is a constant,
  to compute $\bx_{i,k}$ in (\ref{eq:bb}) one needs 
  to go through
$$
\alpha^s=\alpha^{2\alpha^2\cdot \lceil \ln(n/\epsilon)\rceil}=(n/\epsilon)^{O(1)}
$$
many pure strategy profiles of players $(i,j)$, $j\in [s]$.
Thus $\vy$ can be computed in time polynomial~in the number of bits needed
  in the binary representation of $\vx$ and $1/\eps$.
\end{proof}




\section{Proofs of Theorems \ref{thm:main} and \ref{thm:succinct}}
\label{sec:succinct}

We use the query-efficient reduction given above to prove
  \thref{thm:main} and \thref{thm:succinct}.
  
\begin{proof}[Proof of \thref{thm:main}]
  By \thref{thm:bab14}, there exist two constants $\eps'>0$ and $c'>0$ such that
  $$\text{\textsc{QC}}_{p'}(\textbf{WSNE}(n',\epsilon))=2^{\Omega(n')},\quad\text{where $p'=2^{-c' n'}$.}$$
Let $n=8n'\cdot\left\lceil \ln (n'/\eps')\right\rceil$ and
  $\eps=8\eps'$. It follows from \thref{lem:queryefficient} and \thref{lem:WSNtoANE} that
  $$\text{\textsc{QC}}_{p'}(\textbf{ANE}(n,\epsilon))\ge\text{\textsc{QC}}_{p'}(\textbf{WSNE}(n',\epsilon)) =2^{\Omega(n')}.$$
The theorem then follows from  
  $n'=\Omega(n/\log n)$. 
\end{proof}

\begin{proof}[Proof of \thref{thm:succinct}]
From \thref{lem:WSNtoANE}, it suffices to show that, given any $\alpha$-action
  succinct game $\calG=(n,U)$, one can construct, in polynomial time,
  a Boolean circuit $U'$ that implements the payoff functions of players in $\calG'$.  
This can be done by following the definition of $\calG'$ in the previous section,
  as the payoffs of a pure strategy profile $\va$ in $\calG'$ only
  depends (in a straight-forward fashion) on the payoffs of $\alpha n=O(n)$
  easy-to-compute profiles of $\calG$.
\end{proof}






\section{Conclusion}\label{sec:conclude}

In this paper, we present a simple and efficient reduction from the problem of finding a 
  WSNE to that of finding an ANE in large games with a
  bounded number of actions.~Our results complement the existing study on
  relations between WSNE and ANE. As an application, we obtain a lower bound
   on 
the randomized query complexity of $\epsilon$-ANE for some constant $\eps>0$. 
It would be interesting~to see other applications of our
  reduction in understanding the complexity of Nash equilibria.
It also remains an open problem to remove the $\log n$ factor in the
  exponent of our lower bound, i.e. to show that
  the number of queries needed to reach an $\eps$-ANE is indeed $2^{\Omega(n)}$.
This $\log n$ factor shows up because we simulate each player in the original game
  with $O(\log n)$ players in the new game. 
Is there a more efficient simulation that uses fewer players?


\printbibliography

\end{document}